\documentclass[sigconf,natbib=false]{acmart}
\usepackage{listings}
\usepackage{amsthm}
\usepackage{caption}
\usepackage{subcaption}
\newcommand*{\defeq}{\stackrel{\text{def}}{=}}
\usepackage{amsmath}
\AtBeginDocument{%
  }

\setcopyright{acmlicensed}
\copyrightyear{2018}
\acmYear{2018}
\acmDOI{XXXXXXX.XXXXXXX}

\RequirePackage[
  datamodel=acmdatamodel,
  style=acmnumeric,
  ]{biblatex}
\begin{document}

\title{Obligation-Producing Actions in Situation Calculus}

\author{Kalonji Kalala}
\email{hkalo081@uottawa.ca}
\orcid{0009-0002-7617-9337}
\authornotemark[1]
\affiliation{%
  \institution{University of Ottawa}
  \city{Ottawa}
  \state{Ontario}
  \country{Canada}
}

\author{Iluju Kiringa}
\email{iluju.kiringa@uottawa.ca}
\affiliation{%
     \institution{University of Ottawa}
    \city{Ottawa}
    \state{Ontario}
    \country{Canada}
}
\author{Tet Yeap}
\email{tyeap@uottawa.ca}
\affiliation{%
   \institution{University of Ottawa}
  \city{Ottawa}
  \state{Ontario}
  \country{Canada}
}

\renewcommand{\shortauthors}{Kalonji Kalala et al.}

\begin{abstract}
This paper proposes a Situation Calculus solution to the frame problem for obligation-producing actions which are actions that create obligations on the part of the agent that performs them. As an example of such actions, we have an opening door action performed by an agent  which has the subsequent obligation of getting the door closed. 
Demolombe and others extend Raymond Reiter’s solution to the frame problem for ordinary actions to accommodate obligation-producing actions. Obligation-producing actions do affect the truth value of a newly introduced fluent that captures the accessibility relation used in semantics of obligation modalities in the Situation Calculus.  Our work simplifies Demolombe’s characterization of the accessibility relation by eliminating the notion of ideality of situations, thereby remaining close to Kripke-style possible-world semantics for deontic logic, in the spirit of Governatori’s approach. Furthermore, we spell out details of a complete solution by extending basic action theories of Reiter to the new setting. Finally, we extend Reiter's regression operator for reasoning about actions back to the initial situation to this new setting. Our solution yields intuitive properties that one would expect from obligations: for example, if a sentence  is obligatory to an agent in a given situation, it remains so in subsequent situations, unless the obligation is explicitly stopped.
\end{abstract}

\ccsdesc[500]{Do Not Use This Code~Generate the Correct Terms for Your Paper}
\ccsdesc[300]{Do Not Use This Code~Generate the Correct Terms for Your Paper}
\ccsdesc{Do Not Use This Code~Generate the Correct Terms for Your Paper}
\ccsdesc[100]{Do Not Use This Code~Generate the Correct Terms for Your Paper}

\keywords{Situation Calculus, legal contract, Obligations, Frame Problem, Regression, Action, Fluent}


\maketitle

\section{Introduction}  \label{Introduction}
The Situation Calculus presented in~\cite{mccarthy1963situations} is a logical language that has been used in artificial intelligence for specifying and reasoning about dynamical systems such as robotics, database updates, control systems, and simulated software agents. Axioms are provided to capture the prerequisites of actions of the domain and the effects of these actions on the external world around the specified system. The effects of actions are captured by fluents,  predicates whose truth values are changed as the result of the performance of actions. In this context, the challenge of specifying axioms that describe what remains unchanged in a compact way is known as the frame problem in artificial intelligence. The seminal work reported in	 ~\cite{DBLP:conf/birthday/Reiter91} has proposed the so-called successor state axioms as a solution to the frame problem for actions that change the external world. In ~\cite{DBLP:conf/aaai/ScherlL93}, authors extend Reiter's solution to the frame problem to actions that change the state of the knowledge of a reasoning agent; these actions are called {\it  knowledge-producing actions}.  The work in~\cite{DBLP:conf/deon/DemolombeH04} gives the first formulation of Deontic Logic concepts in the Situation Calculus to capture an agent's set of obligations. The authors use the approach to knowledge  from~\cite{DBLP:conf/aaai/ScherlL93} by representing obligations through an appropriate first-order representation of the accessibility relation of the semantics of modal Deontic Logic.

This paper extends the ideas found in~\cite{DBLP:conf/deon/DemolombeH04} to a full-fledged solution. Demolombe et al. did not develop a complete framework stretching from a logic of obligation-producing actions to reasoning about these actions and implementing specifications written in the said logic of obligation-producing actions. 
We go further than the work in ~\cite{DBLP:conf/deon/DemolombeH04} by simplifying the formalization to keep it close to  the Kripke-style possible world semantics for deontic logic from \cite{Hintikka70} that was embedded  in the Situation Calculus by Scherl and Levesque. We simplify Demolombe’s axiom for characterizing the newly introduced accessibility relation by getting rid of the notion of ideality of situations. Furthermore, we spell out details of a complete solution by extending basic action theories of Reiter to the new setting. In addition, we extend Reiter's regression operator for reasoning about actions back to the initial situation to this new setting and we extend his  regression theorem to obligations. Finally, we show that the formalization yields intuitive properties that one would expect from obligations: for example, if obligatory sentences remain so in subsequent situations, unless the obligation is explicitly stopped. 

There are two main motivations for formalizing obligations in the Situation Calculus. The first motivation is of a purely scientific nature: namely, how to formalize logical modalities within the largely first-order framework of the Situation Calculus. Although this has been accomplished for the epistemic modality of knowledge, it remains not yet fully realized for the deontic modality of obligation. The second motivation for formalizing obligations is to enable the precise specification and automated reasoning of \textit{legal documents} and \textit{smart legal contracts}. A legal contract is a legally enforceable agreement that contains requirements for parties to engage in business transactions \cite{DBLP:journals/dss/Lee88}\cite{macleod20002007}. A smart contract is the logical implementation of a legal contract. These artifacts have become increasingly prominent in cyber-physical systems, where they serve as mechanisms for enforcing formally defined (and thus obligatory) agreements among interacting entities. By providing a rigid semantic foundation, formal models of obligations support verification, consistency checking, and the reliable execution of such agreements in dynamic and distributed environments (see, e.g.,~\cite{DBLP:conf/otm/KruijffW17}). 

\section{Formal Background} 
\label{sec:formalPreliminaries}

\subsection{Sequential and Temporal Situation Calculus}\label{SitCalc}

The Situation Calculus \cite{mccarthy1963situations,DBLP:conf/birthday/Reiter91} is a many-sorted and mostly first order language with equality specifically designed for representing dynamically changing world. 
We consider a version of the Situation Calculus with four sorts for actions ($\mathcal{A}$), situations ($\mathcal{S}$), time points ($\mathcal{T}$), 
and objects ($\mathcal{O}$) other that the first three sorts. \\
{\bf Actions} are first order terms consisting of an action function symbol and its arguments, one of which being the action occurrence time. Named actions are the cause of every change in the world. For example : 
\begin{itemize}
	\item \(unlock(d,t)\) : An action that  represents the fact of unlocking the door \(d\) at the time \(t\).
	\item \(lock(d,t)\) : An action that represents the fact of locking the door \(d\) at the time \(t\).
\end{itemize}
 Actions are executed by agents in the domain. \\
\noindent{\bf Situations} are first order terms denoting finite sequences of actions and represented by a binary function symbol $do$: $do(\alpha,s)$ denotes the sequence resulting from adding the action $\alpha$ to a sequence $s$. \(do(\alpha,s)\) denotes the successor situation to $s$ obtained by performing the action $\alpha$. For example 
\begin{itemize}
	\item \( do(unlock(D,10),S_0)) \) : Represents the next situation after executing the action \( unlock(D,10) \) in the situation $S_0$. $D$ is a specific door that is unlocked at point time 10. 
\end{itemize}

\noindent The constant $S_0$ ({\bf initial situation}) denotes the empty sequence $[\;]$ of actions. It is used to indicate the \textit{initial situation}, which means that there no exists a situation before $S_0$. \\
\noindent {\bf Time points}: In the temporal Situation Calculus, the history of the world will be represented with the exact duration of time, or the range of times of the actions. Time will be expressed explicitly as a sequence of real numbers. \\
Finally, {\bf objects} represent domain specific individuals other than actions, situations, and time points.  The language's alphabet has variables and a finite number of constants for each sort, a finite number of action function symbols, a finite number of situation independent functions, a finite number of {\it functional fluents}, which are function symbols with a situation argument, a finite number of situation independent predicates, and a finite number of {\it predicate fluents} with a situation argument. \\
\noindent \textbf{Predicate fluents} or \textit{relational fluents} represent properties whose truth values vary from situation to situation as a consequence of executions of  actions. A predicate fluent is denoted by a predicate symbol whose last argument is a situation term. For example : 
\begin{itemize}
	\item $locked(d,s)$: A relational (predicate) fluent indicating  that the agent has locked the door in the situation $\textit{s}$. The value of the predicate should be \textit{True} or \textit{False}.
\end{itemize}
\noindent \textbf{Functional fluents} denote values that  vary from situation to situation as a consequence of executions of actions.  \\
The language  also  includes special predicates $Poss$, and $\sqsubset$. The predicate  \textit{$\sqsubset$} :\textit{ situation $\times $ situation}  establishes the \textit{relation of order} on the situations. 
$s \sqsubset s'$ states that the situation $s'$ is reachable from the situation $s$ by performing some sequence 
of actions. The predicate symbol \textit{Poss}:  \textit{action $\times $ situation}. \textit{Poss(a,s)} reflects the fact that it is possible to accomplish the action \textit{a} in situation \textit{s}.  \\

In this approach of axiomatizing legal contract domain into the Situation Calculus, we consider fluents to be used  to represent the obligations between parties. The truth value of a such fluent  will determine the performing of an obligation. Some actions called \textit{obligation producing actions} are used to produce obligation whoever agent perform them. In $S_{0}$ all obligations represented by fluents are  false.\\

A dynamic  domain is axiomatized in the Situation Calculus with axioms which describe how and under what conditions change occurs or not as a result of performing actions. Such axioms are called {\it basic action theory} in~\cite{DBLP:conf/birthday/Reiter91}, as expressed in the definition \ref{TheBasicActionsTheoriess}. They include the following classes of sentences: foundational axioms for situations; action precondition axioms stating the conditions of change; successor state axioms stating how change occurs; specific axioms for time, stating the action occurrence times and start times of situations;  unique names axioms for action terms; and axioms describing \(S_0\).
\begin{definition}{\bf The Basic Action Theories}
	\label{TheBasicActionsTheoriess} \end{definition}

A \textit{basic action theories} is 
a set of axioms $ \mathfrak{D}$ in the form expressed by  formula in ~\ref{BasicActionTheories}. The formula designates  a basic action theory: 
\begin{align}
	\begin{split}\label{BasicActionTheories}
		\mathfrak{D} = \Sigma \cup  \mathfrak{D}_{ss} \cup \mathfrak{D}_{ap} \cup \mathfrak{D}_{una} \cup \mathfrak{D}_{S_{o}}
	\end{split}
\end{align}
with : 
\begin{itemize}
	\item $\Sigma $ \textit{contains the set  the four foundational axioms} \cite{LevesquePR98}.
	\item $\mathfrak{D}_{ss} $ \textit{contains  a collection of successor state axioms. }
	\item $\mathfrak{D}_{ap} $ \textit{contains  a collection of actions precondition axioms}.
	\item $\mathfrak{D}_{una}$ \textit{contains unique names axioms for situations }.
	\item $\mathfrak{D}_{S_{o}}$ \textit{consists of a collection of first order sentences which are uniform in $S_{0}$ }.
\end{itemize}

$S_{0}$ represents the \textit{initial state of any domain}. Our domain is about obligations in legal contract, thus we will have \textit{the initial legal contract state}. All sentences of $\mathfrak{D}_{S_{o}}$ only indicate the $S_{0}$ term of sort situation \cite{LevesquePR98}. This because there is not any sentence containing in $D_{S0}$ that indicates any  \textit{do} (functional symbol), \textit{Poss} or $\sqsubset$.

\subsection{Frame Problem}
The frame problem in Artificial Intelligence examines how to represent what remains unchanged after an action without explicitly listing all unaffected facts \cite{hayes1981frame}\cite{brown2014frame}. In formalisms such as Situation Calculus, actions modify some fluents (properties of the world), but the majority of them persist \cite{DBLP:journals/etai/LevesquePR98}\cite{reiter2001knowledge}. Reasoning becomes ineffective when this is naively encoded since it demands a lot of "frame axioms." This is addressed by solutions like successor state axioms and non-monotonic reasoning, which assume persistence in other situations and only specify when changes take place.

In AI systems dealing with law, the frame problem is especially important since legal logic requires tracking how actions alter rights, duties, and rules—without accidentally changing unrelated legal facts. For instance, fulfilling a fine modifies the status of an obligation but shouldn't modify other contracts or entitlements. When using systems built on deontic logic, it's crucial to update only the legally relevant effects while keeping the broader legal framework unchanged. Solving this challenge efficiently is key to building trustworthy legal reasoning agents.
In this paper, we use successor state axioms in Situation Calculus to handle the frame problem to capture changes that happen, and later on,  to be able to build a complete framework to reason about obligations in \textit{multi-agent legal systems} \cite{drumond2008multi}\cite{sperotto2019ontology}.

\subsection{Running Example: Opening a Door}\label{OpeningDoor}  
We consider a variation of the Moore's safe opening example as expressed in~\cite{moore1981reasoning}, with a door replacing a safe. 
To start, let us enumerate some of the actions and fluents of the domain. \\
{\bf Primitive actions} are: 
$unlock(d,t)$,  
$lock(d,t)$, 
$moveTo(d)$, and  
$notify(m,t)$. \\ 
{\bf Fluents} are: $open(d,s)$, $locked(d,s)$, $notifiedManager(s)$, $at(d,s)$. \\
Finally, {\bf situation independent predicates and function} are:  
$manager(m)$, and  
$door(d)$. \\
All these actions, fluents, functions, and predicates are intuitively understandable, except the following: 
\\[.4ex]
{\bf Primitive Actions}.
\begin{itemize}    
	\item $pressButton(d,E,t)$: press the button to open the door $d$ with
	credential $E$ at time $t$; $E$ is a constant meaning "Employee".  
	\item $notify(m,t)$: notify the manager $m$ of the locking of the door at 
	time $t$.    	 	
\end{itemize}  
{\bf Fluents}.
\begin{itemize} 
	\item $locked(d,s)$: relational fluent meaning that 
	the agent has locked the door in $\textit{s}$.
	\item $notifiedManager(s)$: functional fluent meaning the manager has been 
	notified of the locking of the door.   
	\item $at(d,s)$: the agent is at the door $d$ in situation $s$.
	\item $open(d,s):$ relational fluent meaning that the agent has open the door $\textit{d}$ in $\textit{s}$.
\end{itemize} 
{\bf Action Precondition Axioms}. There is one for each action function $A({\vec x,t})$, with syntactic form
\begin{align}\label{APAs}
	\begin{split}
		Poss(&A(\vec x,t),s) \equiv  \Pi_A(\vec x,s). 
	\end{split}
\end{align}
Here, $\Pi_A({\vec x},s)$ is a first order formula with free variables among ${\vec x},s$. Moreover, the formula on the right hand side of (\ref{APAs}) is uniform in s\footnote{A formula $\phi(s)$ is uniform  in a situation term $s$ if $s$ is the only situation term that all the fluents occurring in $\phi(s)$ mention as their last argument.}. The following sentences states the condition under which the primitive actions listed above may be performed: 
\begin{align}  
	\begin{split}\label{APA-unlock}  
		Poss(&unlock(d,t),s) \equiv  door(d)\land at(d,s) \land locked(d,s), 
	\end{split}\\
	\begin{split}\label{APA-lockDoor}
		Poss(&lock(d,t),s)  \equiv 
		open(d,s) \land at(d,s) \land door(d), 
	\end{split}\\
	\begin{split}\label{APA-moveTo}
		Poss(&moveTo(d,t),s)\equiv  true, 
	\end{split} \\
	\begin{split}\label{APA-pressButton}  
		Poss(&pressButton(d,c,t),s) \equiv
		at(d,s)\land door(d)\land c = E, 
	\end{split} \\
	\begin{split}\label{APA-NOTIFY}
		Poss(&notify(m,t),s)\equiv  true.
	\end{split}
\end{align}
{\bf Successor State Axioms}. There is one for each $(n+1)$-ary relational fluent $F$, with the following syntactic form: 
\begin{align}  \label{SSAs}
	\begin{split}
		F({\vec x}, do(a,s)) \equiv  \Phi_F(\vec x,a,s).
	\end{split}
\end{align}
In addition, there is one such axiom for each $(n+1)$-ary functional fluent $f$, with the following syntactic form: 
\begin{align}  \label{SSAs-FF}
	\begin{split}
		f({\vec x}, do(a,s)) = y \equiv   \Phi_f(\vec x,y,a,s).
	\end{split}
\end{align}
The formulas on the right hand sides of (\ref{SSAs}) and (\ref{SSAs-FF}) are uniform in s, and $\Phi_F({\vec x},a,s)$ as well as $\Phi_f({\vec x},y,a,s)$ are formulas with free variables among ${\vec x},a,s$ and ${\vec x},y,a,s$, respectively. The formula $\Phi_F({\vec x},a,s)$ specifies how actions of the domain impact the truth value of a relational fluent $F$ and has the following canonical form ~\cite{reiter2001knowledge}: 
\begin{equation}\label{CANONICAL-BAT}
	\gamma \! {+ \atop F}(\vec x,a,s) \lor 
	F(\vec x,s) \land \neg \gamma \! {- \atop F}(\vec x,a,s),
\end{equation}
where $\gamma \! {+ \atop F}(\vec x,a,s)$ 
($\gamma \! {- \atop F}(\vec x,a,s)$) denotes a first order formula 
specifying the conditions that make a fluent \(F\) true (false) in the 
situation following the execution of \(a\). 
The formula $\Phi_f({\vec x},y,a,s)$ specifies how actions of the domain impact functional fluents; its canonical form is similar to the one for $\Phi_F({\vec x},a,s)$.  
Successor state axioms for the fluents of the door opening domain are as follows: 
\begin{align} 
	\begin{split}\label{SSA-openDoor}
		open(d, do(a,s))& \equiv (\exists t,c)a = unlock(d,t) \land \\  
		&(c = E \land pressButton(d,c,t)) \lor  \\
		&open(d,s) \land a  \not = lock(d,t),
	\end{split}\\  
	\begin{split}\label{SSA-lockedDoor}
		locked(d,do(a,&s)) \equiv 
		(\exists t)a  = lock(d,t) \lor (locked(d,s) \land \\
		&\neg (\exists t', c)(c = E \land a = pressButton(d,c,t'))),
	\end{split} \\
	\begin{split}\label{SSA-at}
		at(d,do(a,s))& \equiv (\exists t) a=moveTo(d,t) \lor \\
		&at(d,s) \land \neg (\exists d',t') moveTo(d',t'),
	\end{split}\\ 
	\begin{split}\label{SSA-notifiedManager}
		notifiedManager&(do(a,s))=m \equiv 
		(\exists t) (manager(m) \land \\
		&a = notify(m,t)) \lor notifiedManager(s).
	\end{split}
\end{align}

\section{Ideality and possible worlds approaches}
While ideality is a normative ordering over those worlds that assesses how well each one satisfies obligations or norms. As a semantic tool, possible worlds explain the various ways the world could be, indicating alternative situations or states of affairs. In possible worlds, every world corresponds to a full description of a situation, including what is true, what happened, what obligations hold, etc \cite{hansson2006ideal}\cite{prakken2015law}.

Some research have been done to introduce the obligation concept of the Standard Deontic Logic (SDL) into the Situation Calculus \cite{DBLP:conf/ecai/Demolombe04} and \cite{DBLP:conf/deon/DemolombeH04}. Authors ranked deontic alternatives in terms of their levels of ideality; they subsequently define the obligatory sentences as those that are true in all alternative situations with maximal ideality; and finally, they give a successor state axiom for the fluent O. By contrast, our work simplifies the formalization by removing any use of situation idealities and by solely embedding the possible world  semantics for SDL from \cite{hintikka1971some} in the Situation Calculus. We use the possible worlds approach the same way Governatori in \cite{DBLP:journals/auasjlog/GovernatoriR06} \cite{DBLP:conf/deon/GovernatoriOCR16} used it to incorporate the possible worlds, typically to give a sequence semantics for norms and obligations. We use the intuition of using possible worlds semantics behind the sequence semantics for norms and obligation, and extend it to the situation Calculus to reason about the obligations.  

\section{Obligations in the Situation Calculus}\label{ObligationInSitCalc}

\subsection{A Deontic Fluent for Expressing Obligations}\label{deonticFluent}

To embed the possible worlds typically used to give a semantics to the Standard Deontic Logic (SDL) from~\cite{vonWright1951-VONAEI-2} into the Situation Calculus, we need a binary {\it deontic accessibility} relation over situations, where some situation $s’$  is seen as being accessible from some other situation $s$ such that, as far as the agent located in situation $s$ is concerned, anything obligatory to that agent in $s$ must be true in situation $s’$. In this way, something being obliged in situation $s$ means that something is true in all situations $s’$ which are deontic accessible situations from $s$. 

Like in~\cite{DBLP:conf/ecai/Demolombe04}, we treat obligation as a fluent by introducing a binary relation $O(s’, s)$, to be read as “$s’$ is deontically accessible from $s$". 
We can now define \textbf{Oblg} as a necessity operator over the                         $O$-accessibility relation:  
\begin{align}\label{OBLG-DEFINITION-1}
	(\forall s'). \textit{\textit{O}} (s',s) \supset \phi\left[ s'\right],   
\end{align}
where $\phi[s']$ is the formula $\phi$ with  situation arguments added recursively to fluents that occur in $\phi$. The following expresses that it is obligatory to have the $door$ locked: 
\begin{equation}  
	\begin{aligned}                        
		(\forall s'). O(s',s) \supset locked(d, s').
	\end{aligned}
\end{equation}
The notation $\textbf{Oblg}(\phi, s)$  says that the formula $\phi$ is obligatory in situation $s$:   
\begin{align}\label{OBLG-DEFINITION-2}
	\textbf{Oblg}(\phi, s)  \defeq (\forall s'). \textit{\textit{O}}
	(s',s) \supset \phi\left[ s'\right].   
\end{align}
The formula $\phi$ used in the abbreviation~(\ref{OBLG-DEFINITION-2})  represents a  formula obtained from a Situation Calculus formula by recursively suppressing its situation arguments. Conversely, $\phi[s]$ represents a Situation Calculus formula obtained by recursively restoring its suppressed situation arguments. This abbreviation~(\ref{OBLG-DEFINITION-2}) as well as the formula (\ref{OBLG-DEFINITION-1}) are based on the semantic condition $(C.O^+)$  given in~\cite{Hintikka70} as Kripke semantics for the obligation modality.  
For example, $\textbf{Oblg}\big(locked(d), s \big)$  expands as follows: 
\begin{equation}  
	\begin{aligned}                     
		\hspace{-.1cm}\textbf{Oblg}(locked(d),s)  \hspace{-.1cm}  \defeq \hspace{-.1cm}
		(\forall s'). O(s',s) \supset locked(d,s').
	\end{aligned}
\end{equation}  

Finally, we introduce an abbreviation to capture the fact that it is obligatory  that $t$ denotes $x$  in all situations: 
\begin{align}\label{OREF-1}
	\textbf{Oref}(t,s) \defeq (\exists x) (\forall s'). O(s',s)\supset t[s'] = x,  
\end{align} 
which can be equivalently expressed as follows: \\
$\textbf{Oref}(t,s) \defeq (\exists x) \textbf{Oblg}(t=x, s)$.

\subsection{Obligation-Producing Actions}\label{ObligationProducingActions}

Among actions of the domain, some do affect what is happening in the world and obligation-producing actions do affect an agent's state of obligations\footnote{Some may do both, as we shall see later.}. We call the later , by reference to knowledge-producing actions introduced by Scherl and Levesque in~\cite{DBLP:conf/aaai/ScherlL93}. At the atomic level, obligation-producing actions are of two kinds: those actions whose effect is to make some (atomic) formula obligatory, and those whose effects is to make the denotation of some term obligatory. 
As an example of the first kind, the ground action  
$unlock(D,10)$  
executed by the agent in situation $S_0$ makes the ground atomic formula  $locked(D,do(unlock(D,10),S_0))$ obligatory.  
In other words, the following sentence is made true by the execution of the action $unlock(D,10)$:
\begin{align}\label{OBLIGATION-TO-LOCK}
	\begin{split}
		\textbf{Oblg}(locked(D),do(unlock(D,10),S_0)).  
	\end{split}
\end{align} 
In our example, $unlock(d,t)$ is an obligation-producing action that creates the obligation for the agent to subsequently get the door locked. The obligation produced is expressed by the ground sentence (\ref{OBLIGATION-TO-LOCK}); that is, by executing the action $unlock(D,10)$ in situation $S_0$, the agent has the obligation to make sure that in some situation $S$ following the situation $do(unlock(D,10),S_0)$, $locked(D,S)$ is true by virtue of an action executed by the agent to make $locked(D,S)$ true. 

In general, we assume that there is a provision of (finitely many) obligation-producing actions $a_{F_i}(\vec x_i,s)$ where $i= 1 \ldots m$, and that for each one of them, there is a  fluent $F_i(\vec x_i,s)$, $i= 1 \ldots m$, of the domain that is made obligatory in situation $do(a_{F_i}(\vec x_i),s)$ upon the execution of $a_{F_i}(\vec x_i,s)$ in situation $s$. It is important to notice the difference between knowledge-producing actions of Scherl and Levesque and obligation-producing actions that are introduced here. The execution of a  knowledge-producing action $a=A_i(\vec x_i)$ which is associated with a fluent $F_i(\vec x_i,s)$ leads to the knowledge of the truth value of $F_i(\vec x_i,s)$, whatever that truth value is, whereas the execution (in situation $S$) of an obligation-producing action  $a=A_i(\vec x_i)$ which is associated with a fluent $F_i(\vec x_i,s)$ leads to the obligation of the positive truth value of  
$F_i(\vec x_i,do(A_i(\vec x_i),S))$.

For the second kind of obligation-producing actions, we assume that 
there is a provision of (finitely many) such actions $a_{f_j}(\vec x_j)$ where $j= 1 \ldots n$, and that for each one of them, there is a functional fluent $f_j(\vec x_j)$, $j= 1 \ldots n$, of the domain whose denotation is made obligatory to the agent.    

As an example of the second kind of obligation-producing actions, we have $lock(d,t)$. By executing the ground action 
$lock(D,20)$ in the ground situation $S$, the agent makes the following ground atomic formula true: 
\begin{equation}\label{OBLIGATION-TO-NOTIFY} 
	\hspace{-.1cm}(\exists m) 
	\textbf{Oblg}(notifiedManager(do(lock(D,20),S)) 
	\hspace{-.1cm} =  \hspace{-.1cm}m).  
\end{equation} 
Thus, $lock(d,t)$ is an obligation-producing action that creates the obligation for the agent to get the manager notified. This obligation is expressed by the sentence (\ref{OBLIGATION-TO-NOTIFY}). 


\section{Solving the Frame Problem}

\subsection{Successor State Axiom for the Fluent $O$}\label{SuccessorStateAxiomForO} 

Solving the frame problem consists in giving successor state axioms for all fluents of the domain as done in~\cite{DBLP:conf/birthday/Reiter91}. So, we need a successor state axiom for the fluent $O$. 

Suppose the agent executes the action $moveTo(D,20)$ in $S_0$ where no obligation holds and $locked(D,S_0)$ is true. Then, the sentence $at(D,do(moveTo(D,20),S_0))$ holds, and no new obligation is introduced. Now, if $unlock(D,30)$ is executed at time t = 30, the sentence $open(D,do(unlock(D,30),do(moveTo(D,20),S_0,S_0)))$ holds and the agent has the obligation of subsequently locking the door. Finally, the execution of $lock(D,40)$ will stop the obligation for the agent to get the door locked.

The above consideration leads to three sorts of actions. The first sort is made of a provision of ordinary actions that do not produce any obligation with respect to the agent and do not release any existing obligations. The second sort is made of ordinary actions that do not produce any obligation, but they stop existing obligation. Finally, the third sort is made of obligation-producing actions.

To start, suppose that a deontic agent is located in a situation $s$. We can imagine several infinitely many situations $s'_1, s'_2, s'_3, \ldots$, which are deontic alternatives to $s$. Furthermore, suppose that the deontic agent performs some action $a$ in $s$ and therefore lands in the successor situation $do(a,s)$. We now wonder what are the deontic alternatives to $do(a,s)$, and how these alternatives are related to the situations $s'_1, s'_2, s'_3, \ldots$. We must come up with successor state axioms for the three sorts of actions identified above by spelling out how an action $a$ affects the fluent $O$. 

\subsubsection{Non Obligation-Producing Actions}

Consider the case of non obligation-producing actions that only  change the truth value of some fluent, without introducing new obligations nor stopping existing ones.  In this case, the deontic alternatives $s'_1, s'_2, s'_3, \ldots$ to $do(a,s)$ will be related to the deontic alternatives $s^*_1, s^*_2, s^*_3, \ldots$ to $s$ as if the action $a$ would have been performed in the later situations, such that $s^*_i = do(a,s'_i)$, for $i = 1, 2, 3, \ldots$. In summary we have\footnote{We assume that the situations involved in successor state axioms are executable in the following sense:
	\begin{align*}
		\begin{split}
			exec(s) =_{df} (\forall a,s'). do(&a,s') \sqsubseteq s 
			\supset Poss(a,s') 
			\land start(s') \leq time(a).
		\end{split}
	\end{align*}
	Here, we assume axioms that specify the occurrence time $time(a)$ of actions and the start time $start(s)$ of situations.}: 
\begin{align}\label{SSAforO-NOPA}
	O(s', do(a,s))\equiv  (\exists s^*). O(s^*,s) \land   s' = do(a,s^*).
\end{align}
This means that these ordinary actions do not introduce new obligations to those that are already in place in the situation $s$ and in all the situations that are $O$-accessible to the situation $s$: the situation $s'$ which is $O$-accessible to $do(a,s)$ will be a mere reflection of the situation $s^*$ which is $O$-accessible to $s$. 
Figure ~\ref{fig:SSAFor-NOPA} illustrates the sentence (\ref{SSAforO-NOPA}).

In the door opening example, the action $moveTo(d,t)$ is of this sort. For this action, we have the following characterization for $O$:

\begin{align}
	\begin{split}
		O(s', do(&moveTo(d,t),s))\equiv \\
		&(\exists s^*).  O(s^*,s) \land s' = do(moveTo(d,t),s^*). 
\end{split}\end{align}  
In essence, this characterization of the fluent $O$ does not differ from the one for non-knowledge-producing actions of Scherl and Levesque.  



\begin{figure}[t]   
	\centering
	\includegraphics[width=0.45\textwidth]{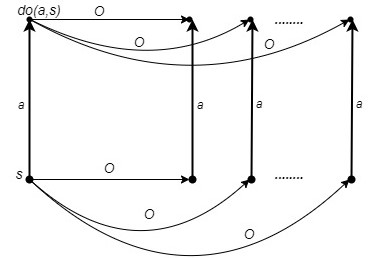}
\caption{Situations that are $O$-related to $do(a, s)$ when 
	$a$ is not an obligation-producing action}
\label{fig:SSAFor-NOPA}
\end{figure} 

\begin{figure}[t]   
	\centering
	\includegraphics[width=0.45\textwidth]{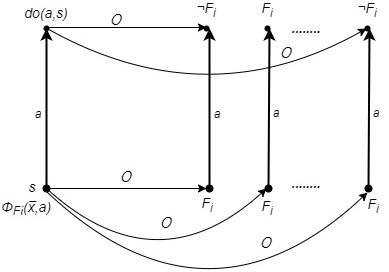}
	\caption{Situations that are $O$-related to $do(a, s)$ when 
		$a$ is an Obligation-Releasing action}
	\label{fig:OBLIGATION-STOPPAGE-NEG}
\end{figure}


\subsubsection{Obligation-Releasing Actions}\label{ObligationReleasingActions}

We now turn to the case of non obligation-producing actions (which we call Obligation-Releasing Actions) that  change the truth value of those fluents of the domain which are associated with obligation-producing actions, and, by doing so, they stop some existing obligations. In our door opening scenario, $notify(m,t)$ is such an action. Consider the reasoning we started with at the beginning of the present subsection. Suppose the agent is in the following situation with specified time  \(S= do(lock(D,40), do(unlock(D,30),do(moveTo(D,20),S_0)))\), where $(\exists m) \textbf{Oblg}(notifiedManager=m, s)$ becomes true. At this point the agent must ensure to notify some manager $M$ in some subsequent situation and at some time, e.g., 50, that the door is closed. This happens when the agent executes the action $notify(M,50)$. So, we reach the situation $do(notify(M),S)$, in that situation the obligation formula :  $(\exists m) \textbf{Oblg}(notifiedManager=m, s)$ ceases to be true. 

For any action $a$, the termination of the obligation of $F_i(\vec {x_i})$ for the agent in situation $do(a,s)$ happens when $F_i(\vec {x_i}, do(a,s))$ holds, which in turn happens when the formula $\Phi_{F_i}(\vec {x_i},a,s)$ 
becomes true. This justifies that we must produce a successor state axiom for $O$ which must entail $m$ sentences of the following form, one for each of the $m$ fluents $F_i(\vec x_i)$ that are associated with the obligation-producing actions $a_{F_i}(\vec {x_i})$, $1 \leq i\leq m$:
\begin{equation}\label{OBLIGATION-STOPPAGE-NEG} 
	\begin{split}
		(\forall s')[O(&s',s)\supset F_i(\vec{x_i},s')] \land \\
		&\Phi_{F_i}(\vec {x_i},a,s) 
		\supset \neg (\forall s^*)[O(s^*,do(a,s))\supset F_i(\vec{x_i},s^*)]. 
	\end{split} 
\end{equation} 
Figure~\ref{fig:OBLIGATION-STOPPAGE-NEG} illustrates the sentence (\ref{OBLIGATION-STOPPAGE-NEG}). The figure shows that the situations that are $O$-alternatives to $do(a,s)$, where $a$ is an obligation-releasing action associated with $F_i$ and $\Phi_{F_i}(\vec {x_i},a,s)$ is true, are the images of   $O$-alternatives to $s$ under the performance of $a$ in which we select those images where $F_i$ is false.

An argument similar to the case of obligation termination can be made for obligation persistence. 
Sentences~(\ref{OBLIGATION-STOPPAGE-POS}) give the sufficient conditions for the  obligations's persistence:
\begin{equation}\label{OBLIGATION-STOPPAGE-POS} 
	\begin{split}
		(\forall s')[O(&s',s)\supset F_i(\vec{x_i},s')] \land \\ 
		&\neg \Phi_{F_i}(\vec {x_i},a,s) 
		\supset (\forall s^*)[O(s^*,do(a,s))\supset F_i(\vec{x_i},s^*)]. 
	\end{split} 
\end{equation} 
Again, one such sentence must be entailed by the successor state axiom for
$O$ for each of the $m$ fluents $F_i(\vec x_i)$ that are associated 
with the obligation-producing actions $a_{F_i}(\vec {x_i})$, $1 \leq i\leq m$.   

\subsubsection{Obligation-Producing Actions}


We specify how an obligation-producing action $a_{F_i}(\vec{x_i})$ that is associated with a fluent ${F_i}(\vec{x_i})$ affects the fluent $O$. Recall that the action $a_{F_i}(\vec{x_i})$, which is executed in a situation $s$, is understood to make the positive truth value of $F_i(\vec{x_i})$ obligatory in the successor situation $do(a_{F_i}(\vec{x_i}), s)$. Thus, in this case, we want that $\textbf{Oblg}(F_i(\vec{x_i}), do(a_{F_i}(\vec{x_i}), s))$ holds.  
This leads to the following characterization of the fluent $O$: 
\begin{align}\label{SSAforOPA-1}
	\begin{split}   
		O(s', &do(a_{F_i}(\vec{x_i}),s))\equiv  \\
		&(\exists  s^{*},a). O(s^{*},s)\land 
		s' = do(a(\vec{x_i}),s^{*}) \land  F_i(\vec {x_i},do(a,s^{*})).  
	\end{split}
\end{align}  
Using the successor state axiom (\ref{SSAs}) and the canonical form  (\ref{CANONICAL-BAT}) of its right-hand side, we get the following:
\begin{align}
	\begin{split}
		F_i(\vec {x_i},do(a,s^{*}))\equiv \Phi_{F_i}(\vec {x_i},a,s),
	\end{split}
\end{align} 
which brings the sentence (\ref{SSAforOPA-1}) to become:
\begin{align}\label{SSAforOPA-2}
	\begin{split}   
		O(s', &do(a_{F_i}(\vec{x_i}),s))\equiv \\
		&(\exists s^{*},a'). O(s^{*},s)\land 
		s' = do(a'(\vec{x_i}),s^{*}) \land \Phi_{F_i}(\vec {x_i},a',s). 
	\end{split}
\end{align}   
Sentence (\ref{SSAforOPA-2}) tells us that the situations that are considered $O$-related  to $do(a_{F_i}(\vec{x_i}),s)$ are those obtained by performing some (non obligation-producing) action $a'$ in situations $O$-related to $s$ in which  
$\Phi_{F_i}(\vec {x_i},a',s)$ evaluates to $true$. 
Referring to sentence (\ref{SSAforOPA-2}), we obtain the following sentence for each obligation-producing action $a_{F_i}(\vec{x_i})$, $1\leq i \leq m$: 
\begin{align}\label{SSAforOPA-4}
	\begin{split} 
		a= & a_{F_i}(\vec{x_i}) \supset \\
		[&O(s', do(a(\vec{x_i}),s))\equiv \\
		&\hspace{0.2cm}(\exists s^{*},a'). O(s^{*},s)\land 
		s' = do(a'(\vec{x_i}),s^{*}) \land  
		\Phi_{F_i}(\vec {x_i},a',s^*)].      
	\end{split}
\end{align}
Notice that an obligation-producing action $a_{F_i}$ executed in $s$, is  supposed to make the $F_i$ obligatory in $do(a_{F_i},s)$. Henceforth, for any situation $s^*$ such that $\neg F_i(s^*)$  holds, $s^*$ must not be considered as a doxatic alternative to $do(a_{F_i},s)$. Those situations where $F_i(s^*)$  holds will be retained as the situations that are $O$-accessible to $do(a_{F_i},s)$: they are exactly the images, under the performance of action $a'$, of the $O$-accessible situation to $s$, where action $a'$ is the one free action variable mentioned in  $\Phi_{F_i}(\vec {x_i},a',s^*)$.   Figure~\ref{fig:SSAforOPA-4} illustrates the sentence (\ref{SSAforOPA-4}).
\begin{figure}[t] 
	\centering
	\includegraphics[width=0.50\textwidth]{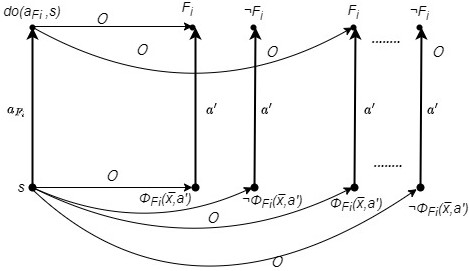} 
	\caption{Situations that are $O$-related to $do(a, s)$ when 
		$a$ is an obligation-producing action producing an obligation 
		$F_i$}
	\label{fig:SSAforOPA-4}
\end{figure}


Given the provision of (finitely many) obligation-producing actions $a_{F_i}(\vec x)$ where $i= 1 \ldots m$, assume that they are associated with relational fluents $F_i(\vec x)$, $i= 1 \ldots m$. Moreover, given a further provision of (finitely many) obligation-producing actions $a_{f_j}(\vec x)$ where $j= 1 \ldots n$, assume that they are associated functional fluents $f_j(\vec x)$, $j= 1 \ldots n$. We are now in position to draw the final conclusion about the characterization of the fluent $O$:  the sentences~(\ref{SSAforO-NOPA}), (\ref{OBLIGATION-STOPPAGE-NEG}), (\ref{OBLIGATION-STOPPAGE-POS}), and (\ref{SSAforOPA-4}) can be put together in the equivalent sentence (\ref{SSAforOPA-5-2}) (that takes into account functional fluents) shown in Figure~\ref{fig:SSAforOPA-5-2}. Sentence (\ref{SSAforOPA-5-2}) is our proposed  successor state axiom  for $O$.   
\begin{figure}[!t]
	\centering 
	\caption{Successor state axiom for the fluent  $O$}
	\label{fig:SSAforOPA-5-2}
	\hrule 
	\begin{align}\label{SSAforOPA-5-2}
		\begin{split} 
			O(s', & do(a,s))\equiv 
			[(\exists s^*,a'). O(s^*,s)\land s' = do(a',s^*) \land \\
			[[&[a\neq a_{F_1}(\vec{y_1})\land \ldots \land 
			a\neq a_{F_m}(\vec{y_m}) \land  \\
			&a\neq a_{f_1}(\vec{z_1})\land \ldots \land 
			a\neq a_{f_n}(\vec{z_n}) \land a = a'] \land\\
			&(\forall \vec{x_1})[(F_1(\vec{x_1},s^*) \land \Phi_{F_1}(\vec {x_1},a,s)\supset \neg  F_1(\vec{x_1},s')] \land \ldots \land\\
			&(\forall \vec{x_m})[(F_m(\vec{x_m},s^*) \land \Phi_{F_m}(\vec {x_m},a,s)\supset \ \neg  F_m(\vec{x_m},s')]  \land \\
			&(\forall \vec{x_1})[(F_1(\vec{x_1},s^*) \land \neg \Phi_{F_1}(\vec{x_1},a,s)\supset F_1(\vec{x_1},s')] \land \ldots \land\\      
			&(\forall \vec{x_m})[(F_m(\vec{x_m},s^*) \land \neg \Phi_{F_m}(\vec{x_m},a,s)\supset F_m(\vec{x_m},s')]\land\\
			&(\forall \vec{x_1}, y_1)[(f_1(\vec{x_1},s^*)=y_1 \land \\
			&\hspace{2cm}\Phi_{f_1}(\vec {x_1},y_1,a,s)\supset f_1(\vec{x_1},s') \not= y_1]\land \ldots \land\\
			&(\forall \vec{x_n}, \vec{y_n})[(f_n(\vec{x_n},s^*)=\vec{y_n} \land \\ 
			&\hspace{2cm}\Phi_{f_n}(\vec{x_n},\vec{y_n},a,s)\supset 		                                                      
			f_n(\vec{x_n},s')\not= \vec{y_n}]  \land \\
			&(\forall \vec{x_n}, \vec{y_1})[(f_1(\vec{x_1},s^*)= \vec{y_1} \land \\
			&\hspace{2cm}\neg \Phi_{f_1}(\vec{x_1},\vec{y_1},a,s)  \supset	           
			f_1(\vec{x_1},s') = \vec{y_1}]\land \ldots \land\\      
			&(\forall x_n, y_n)[(f_n(\vec{x_n},s^*)=\vec{y_n} \land \\
			&\hspace{2cm}\neg \Phi_{f_n}(\vec{x_n},\vec{y_n},a,s) \supset        f_n(\vec{x_n},s')= \vec{y_n}]]\lor \\
			& [(\forall x_1)[a=a_{F_1}(\vec{x_1}) 
			\supset  \Phi_{F_1}(\vec {x_1},a',s^{*})] 
			\land \ldots \land\\ 
			&\hspace{.1cm}(\forall x_m)[a=a_{F_m}(\vec{x_m}) \supset 
			\Phi_{F_m}(\vec {x_m},a',s^{*})]\land\\
			& (\forall \vec{y_1}, \vec{y_1})[a=a_{f_1}(\vec{x_1})\supset \Phi_{f_1}(\vec {x_1},\vec{y_1},a',s^{*})]\land\ldots\land\\ 
			&(\forall x_n, \vec{y_n})[ a=a_{f_n}(\vec{x_n}) \supset 
			\Phi_{f_n}(\vec {x_n},\vec{y_n},a',s^{*})]]]]. 
		\end{split}
	\end{align}	
	\hrule
\end{figure}  
The sentence ~(\ref{SSAforOPA-5-2-Example}) in Figure~\ref{fig:SSAforOPA-5-2-Example} is the successor state axiom for $O$ in our running example.   
\begin{figure}[!t]
	\centering 
	\caption{Example of Successor state axiom for the fluent $O$}
	\label{fig:SSAforOPA-5-2-Example}
	\hrule 	\begin{align}\label{SSAforOPA-5-2-Example}
		\begin{split} 
			O(s', &do(a,s))\equiv [(\exists s^*,a'). O(s^*,s)\land s' = do(a',s^*) \land \\
			[[&[a\neq unlock(d,t)\land a\neq lock(d,t)\land a = a'] \land\\
			&(\forall d)[(locked(d, s^*)\land ((\exists t)(a = lock(d,t) \lor \\
			& \hspace{0.3cm} locked(d,s) \land \neg (\exists c,t') a = pressButton(d,c,t'))]\supset \\
					& \hspace{0.3cm}\neg  locked(d,s')] \land\\ 
			&(\forall d)[(locked(d, s^*)\land \neg((\exists t)(a = lock(d,t) \lor locked(d,s) \land \\
			&\hspace{0.3cm}\neg (\exists c,t') a = pressButton(d,c,t'))]\supset locked(d,s')] \land\\  			 
			&(\forall m)[(notifiedManager(s^*)=m \land 
			((\exists t') (manager(m,s) \land \\
			& \hspace{0.3cm}a = notify(m,t')) \lor notifiedManager(s) \land \\ 
			&\hspace{0.3cm}\neg(\exists m',t'') a=notify(m',t''))\supset  \\
			&\hspace{3cm} notifiedManager(s') \not= m]\land  \\                     
			&(\forall m)[(notifiedManager(s^*)=m \land 
			\neg  ((\exists t') (manager(m,s) \land \\
			&\hspace{0.3cm}a = notify(m,t'))  \lor notifiedManager(s) \land \\
			&\hspace{0.3cm}\neg(\exists m',t'') a=notify(m',t''))\supset \\  
			&\hspace{3cm}notifiedManager(s') = m]]\lor \\
			& [(\forall d,t)[a=unlock(d,t)\supset ((\exists t')a' = lock(d,t') \lor \\
			&\hspace{0.3cm}locked(d,s^*) \land a' \not = pressButton(d,t'))]\land \\  
			& [(\forall d,t)[a=lock(d,t)\supset ((\exists m,t') (manager(m,s^*) \land   \\
			&\hspace{0.3cm}a' = notify(m,t')) \lor notifiedManager(s^*) \land  \\ 
			&\hspace{0.3cm}\neg (\exists m',t'')a'=notify(m',t''))]]]].
		\end{split}
\\	\end{align}	
	\hrule 
\end{figure}  

\subsection{General Successor State Axiom for $O$}\label{General-SSA-for-O}

The successor state axiom for $O$ axiom (\ref{SSAforOPA-5-2}) assumes that obligations are atomic fluents. We need to generalize this axiom by considering obligations that are general formulas. Assume that we have a provision of obligation-producing actions $a_{\psi_i}(\vec x_i)$, $i= 1 \ldots m$, each associated with an arbitrary formula $\psi_i(\vec x_i)$, $i= 1 \ldots m$.  These formulas $\psi_i(\vec x)$ are the so-called "objective situation-suppressed sentences" ("Objective sentences" for short) of Reiter  ~\cite{reiter2001knowledge}: they are about the world and do not mention the abbreviation $\textbf{Oblg}$.
Furthermore assume, without loss of generality, that $\psi_i(\vec x_i)$ is the conjunction $G_1(\vec x_1)\land \ldots \land G_k(\vec x_k)$   of atomic formulas $G_i(\vec x_i)$,  $i= 1 \ldots k$. Also, assume that 
$\mathcal{\Phi}_{\psi_i}(\vec {y_i},a,s)$ is the conjunction 
$\Phi_{G_1}(\vec {x_1},a_1,s_1)\land \dots \land \Phi_{G_l}(\vec {x_l},a_1,s_1)$    
of atomic formulas $\Phi_{G_i}$, $i= 1 \ldots l$. Here the formula 
$\Phi_{G_i}(\vec {x_i},a_i,s_i)$ represents the right-hand side of the successor state axiom for the fluent $G_i(\vec {x_i})$.  
We can generalize the axiom (\ref{SSAforOPA-5-2}) by using the general formulas  $\psi_i$ instead of atomic fluents: we omit this here for lack of space. 	

\subsection{Correctness}\label{Properties}

The formalization is correct in the sense that, an obligation-producing action creates an obligation in the successor situation and if a fluent is obligatory in a given situation, it remains so, unless an action that causes the obligation to stop is executed.  
\begin{theorem}{\bf (Correctness)}
	\label{OBLIGATION-PERSISTENCE-AND-STOPPAGE} 
	For all actions $\alpha$, fluents $ F_{i}$, and situations $s$ the following holds:\\[.1cm]
	(1) ({\bf Obligation Creation}): If $\alpha = a_{F_{i}}({\vec x}_{i})$, where $a_{F_{i}}({\vec x}_{i})$ 
	is the obligation-producing action associated with 
	$F_{i}$, then the following
	$ \textbf{Oblg}(F_i({\vec x}_{i}), do(\alpha, s))$ holds.\\[.1cm]
	(2) ({\bf Obligation Stoppage}): If $\textbf{Oblg} (F_{i}({\vec x}_{i}), s)$ and 
	$\Phi_{F_i}(\vec {x_i},\alpha,s)$ hold, then 
	$ \neg \textbf{Oblg}(F_i({\vec x}_{i}), do(\alpha, s))$ holds. \\[.1cm]
	(3) ({\bf Obligation Persistence}): If $\textbf{Oblg} (F_{i}({\vec x}_{i}), s)$ and $\neg  \Phi_{F_i}(\vec {x_i},\alpha,s)$ hold, then 
	$ \textbf{Oblg}(F_i({\vec x}_{i}), do(\alpha, s))$ holds. 
\end{theorem} 
\begin{proof} 
	{\bf Proof}:\\[.2cm]
	(1): For lack of space, we omit the proof. \\[.1cm]
	(2): Suppose that, for some ground object tuple ${\vec X}_{i}$, some ground action $A$, and some ground situation $S$, $\Phi_{F_i}(\vec {x_i},A,S)$ holds. Moreover, assume that we execute some ground action $A(\vec X)$ in $S$. Then, by the successor state axiom for $F$, $F(\vec X_i)$ holds in the subsequent situation $do(A({\vec X}),S)$.  Furthermore, by assumption we have that $\textbf{Oblg} (F_{i}({\vec X}_{i}), S)$ holds; hence, by the abbreviation (\ref{OBLG-DEFINITION-2}), we get the following:
	\begin{align} \label{OBLG-DEFINITION-1-2}
		(\exists s^*)\textit{\textit{O}}(s^*,S)) \supset F_i({\vec X},s^*).   
	\end{align}   
	We must now show that $F_i({\vec X},S)$ ceases to be obligatory in $do(A(\vec X_i), S$; that is, we must show that    
	$ \neg \textbf{Oblg}(F_i({\vec X}_{i}), do(A({\vec X}_{i}), S))$ holds. 
	To do so, suppose that we have some ground situation $S'$ such that $O(S',do(A({\vec X}),S))$. Then, by Axiom~(\ref{SSAforOPA-5-2}) and by the fact that $A$ is not an obligation-producing action, the following sentence holds (after appropriate skolemization and simplifications):  
	\begin{align} \label{SSAforOPA-5-2-2}
		\begin{split}  
			&O(S^*,S)\land S' = do(A,S^*) \land \\
			&(\forall x_i)[(F_i(\vec{x_i},S^*) 
			\land \Phi_{F_i}(\vec {x_i},A,S)\supset \\
			&\hspace{3.9cm}\neg  F_i(\vec{x_i},S')]. 
		\end{split}
	\end{align}	
	By the sentence (\ref{OBLG-DEFINITION-1-2}) and the assumption 
	$\Phi_{F_i}(\vec {X_i},A,S)$, the above sentence  (\ref{SSAforOPA-5-2-2}) yields the following:  
	\begin{align} \label{SSAforOPA-5-2-3}
		\begin{split}  
			O(S^*,S)\land S' = do(A,S^*) \land 			 
			\neg  F_i(\vec{X_i},S')]. 
		\end{split}
	\end{align}
	Since the sentence (\ref{SSAforOPA-5-2-3}) holds for any fixed situation $S$ and $S^*$, we conclude that 
	$ \neg \textbf{Oblg}(F_i({\vec x}_{i}), do(\alpha, s))$ holds. \\[.1cm]
	(3): The arguments of the proof are similar to the proof of obligation stoppage. We omit the proof.  
\end{proof}


\section{Characterization of $O$ in the Initial Situation}\label{INIT-O}


Characterizing the fluent $O$ in the initial situation is important, since, as it stands, the definition of obligation given in (\ref{OBLG-DEFINITION-2}) does not differentiate the obligation modality from knowledge. Such a differentiation comes only in the form of appropriate properties which are unique to the deontic accessibility relationship.  
The space of situations is constrained by foundational axioms, which impose a structuring of that space in the form of a tree rooted in the initial situation \(S_0\) ~\cite{DBLP:conf/birthday/Reiter91}. To accommodate the introduction of knowledge, this $S_0$-rooted tree semantics of the space of situations has been modified in~\cite{DBLP:conf/aaai/ScherlL93} with the introduction of the predicate $Init(s)$ to capture initial situations for knowledge; and the specification given for this predicate still holds for obligations. The next abbreviation~(\ref{DEFINITION-INIT}) and two axioms adapt Reiter's version of the predicate $Init(s)$  ~\cite{reiter2001knowledge} to obligations:

\begin{align}
	&Init(s) \defeq \neg (\exists a,s') s=do(a,s'), \label{DEFINITION-INIT} \\
	&O(s,s') \supset (Init(s) \equiv Init(s')), \label{O-RELATING-INIT-SIT} \\
	\begin{split}\label{WEAKER-INDUCTION-AXIOM}
		&(\forall P). (\forall s)(Init(s) \supset P(s)) \land 
		(\forall a,s)(P(s) \supset \\
		&\hspace{4.2cm} P(do(a,s))) \supset (\forall s)P(s).
	\end{split} 
\end{align} 
For deontic logic, the following restrictions are imposed to the predicate $Init(s)$ to obtain different variants of SDL:  \\[.7ex]
{\bf Secondary Reflexivity}: 
\begin{equation}\label{SECONDARY-REFLEXIVITY}
	\begin{aligned}
		(\forall s, s'). Init(s) \land  Init(s')\supset 
		(O(s',s) \supset O(s',s')).
	\end{aligned}
\end{equation}  
{\bf Seriality}: 
\begin{equation}\label{SERIALITY}
	(\forall s) . Init(s) \supset 
	((\exists s'). Init(s') \land O(s',s)).
\end{equation} 
{\bf Secondary Seriality}: 
\begin{equation}\label{SECONDARY-SERIALITY}
	\begin{aligned}
		(\forall s).Init(s) \hspace{-.1cm}\supset \hspace{-.1cm}
		((\exists s'). Init(s') \land O(s',s)     
		\land O(s',s')).
	\end{aligned}
\end{equation}     
The properties above are used as restrictions to model various deontic logic systems. For example, von Wright's SDL is obtained by using Secondary Reflexivity and Seriality. 
\begin{theorem}({\bf Restriction Theorem}) \label{RESTRICTION-THEOREM}
	Suppose that the properties of secondary reflexivity, secondary seriality, and seriality for the $O$ relation hold in all situations $s$ such that $Init(s$). Then these properties also hold for the $O$ relation in all executable situations. 
\end{theorem}
\begin{proof}~~\\
	We use the induction  principle for the situation calculus  ~\cite{DBLP:journals/jacm/PirriR99} as given in  its weaker form in formula~(\ref{WEAKER-INDUCTION-AXIOM}). So we must show that the properties ~(\ref{SECONDARY-REFLEXIVITY}),  ~(\ref{SERIALITY}), and  ~(\ref{SECONDARY-SERIALITY}) hold in the initial situations, and that, if they hold in situation $s$, then they also hold in executable situations $do(a,s)$. Assume that they hold in initial situations. Thus, we only need to show that whenever they hold in $s$, then they also hold in $do(a,s)$. ~~\\[.2cm]
	{\bf Secondary Reflexivity}: 
	Suppose that 
	\((\forall s,s'). O(s',s) \supset O(s',s')\) 
	holds for $s$. We must prove that 
	\((\forall s,s'').O(s'',do(a,s)) \supset O(s'',s'')\) 
	holds for every $a$ such as $do(a,s)$ is executable. For the sentence \((\forall    s,s'').O(s'',do(a,s)) \supset O(s'',s'')\) to be false there should be some $s$ and some $a$, such that \(O(s'',do(a,s))\) is true, but \(O(s'',s'')\) is false. According to the successor state  axiom (\ref{SSAforOPA-5-2}) for \textit{O}, the formula \(O(s'',do(a,s))\) is true iff \(s'' = do(a,s')\) for some $s'$, and \(O(s',s)\) as well as further conditions on the right-hand side of the axiom (\ref{SSAforOPA-5-2}) are true. By assumption, if \(O(s',s)\) is true, then \(O(s',s')\) must be true as well. Therefore \((\forall s,s''). O(s'',do(a,s)) \supset O(s'',s'')\) cannot be false.\\[.1cm]  
	{\bf Seriality}: Assume that 
	\((\forall s)(\exists s') O(s',s)\) 
	holds. We must show that 
	\((\forall s)(\exists s'')O(s'',do(a,s))\)
	holds for every $a$ such as $do(a,s)$ is executable. For  
	$(\forall s)(\exists s'')O(s'',do(a,s))$ to be false, we will have the case where, for some $s$, no \(s''\) is such that \(O(s'',do(a,s))\) is true. Notice that by the successor state axiom (\ref{SSAforOPA-5-2}) for \textit{O}, the unique alternative for \(O(s'',do(a,s))\) to be true is for \(s''\) to be equal to \(do(a,s')\) for some \(s'\) such that \(O(s',s)\), as well as further conditions on the right-hand side of the axiom (\ref{SSAforOPA-5-2}) to hold. Then, according to the successor state axiom (\ref{SSAforOPA-5-2}), and the assumption that \(O(s',s)\) is true, it has to be the case that \(O(s'',do(a,s)\) is true. Therefore \((\exists s') O(s',s) \supset (\exists s'')O(s'',do(a,s)) \), which contradicts our assumption. \\[.1cm]
	{\bf Secondary Seriality}: 
	Suppose that 
	\((\forall s)(\exists s')(O(s',s) \wedge O(s',s'))\) is true for all $s$. We must prove that 
	\((\forall s)(\exists s'')(O(s'',do(a,s)) \wedge O(s'',s''))\)  
	holds for every \textit{a} such as \(do(a, s)\) is executable. The rest of the proof goes by contradiction. We leave details out. 
\end{proof}
\section{Reasoning About Obligations}\label{Reasoning}
\subsection{Regression}\label{RegressionDefinition}
Regression is a type of proof theory that uses \textit{backward reasoning}. Regression is a mechanism for transposing reasoning about the truth value of formulas in successor situations to reasoning about modified formulas in the initial situation.  

In the Situation Calculus, regression is the most crucial theorem-proving tool. In the type of regression mechanism proposed by Reiter in \cite{reiter1991frame}, the reasoning about future situations is simply  brought to the reasoning about the initial situation \(S_0\) \cite{pirri1999some}.  The outcome of the regression is a formula in ordinary modal logic, meaning that it has no action terms and only has the situation's term \(S_0\).

In a regressable sentences, each situation term is rooted at \(S_{0}\), so one can count the number of actions involved by examining the term. Reiter’s regression operator for reasoning about actions back to the initial situation is a reasoning mechanism for this setting as well.
\subsection{Regressing Obligation Formulas}
We start by defining a formula $W$ of the Situation Calculus to be {\it regressable} iff it mentions only situation terms that are rooted in $S_0$, does not quantify over situations, and does not mention the predicate symbol $\sqsubset$, nor equality atoms over situation terms. 
Suppose we have a regressable formula $W$ where all situation terms have been suppressed. Then the one-step regression of $W$ through the action $\alpha$ , denoted  $\rho^1(W,\alpha)$, is the following manipulation~\cite{reiter2001knowledge}: first take the formula $W$ and restore the situation argument $do(\alpha,\sigma)$ to all fluents mentioned in $W$, for some $\sigma$; then, regress the resulting formula; finally suppress all situations from the resulting regressed formula.   
The following Definition~\ref{OBLIGATION-REGRESSION-OPERATOR} accommodates formulas that mention $\textbf{Oblg}$: 
\begin{definition}{\bf (Regression Operator for Obligations)} 
	\label{OBLIGATION-REGRESSION-OPERATOR}     
	Suppose $W$ is a situation-suppressed (regressable) formula of the Situation Calculus, and 	$\mathcal{D}$ is a basic action theory. Moreover, assume, without loss of generality, that $W$ is the conjunction $G_1(\vec x_1)\land \ldots \land G_k(\vec x_k)$ of atomic formulas 
	$G_i(\vec x_i)$,  $i= 1 \ldots k$, and that $\mathcal{\Phi}_{W}$ is the conjunction $\Phi_{G_1}\land \dots \land \Phi_{G_l}$    
	of atomic formulas $\Phi_{G_i}$, $i= 1 \ldots l$. Here the formula 
	$\Phi_{G_i}$ represents the right-hand side of the successor state axiom for the fluent $G_i(\vec {x_i})$. 
	Then, in addition to the regression steps described in ~\cite{reiter2001knowledge} for regressable formulas, the following steps accommodate  expressions of the form 
	$\textbf{Oblg}(W,do(\alpha,\sigma))$, 
	where $\alpha$ is any of actions mentioned in Section~\ref{SuccessorStateAxiomForO}:  
	\begin{description}
		\item[(i)] An obligation in $S_0$ stops the regression process:
		\[{\mathcal{R}}[\textbf{Oblg}(W,S_0)] = \textbf{Oblg}(W,S_0).\]
		\item[(ii)] If $\alpha$ is a non obligation-producing action that is 
		a non obligation-releasing one, then 
		\[{\mathcal{R}}[\textbf{Oblg}(W,do(\alpha,\sigma))] = 
		{\mathcal{R}}[\neg \mathcal{\Phi}_{W} \land
		\textbf{Oblg}(\rho^1(W,\alpha),\sigma)]. \]   
		\item[(iii)] If $\alpha$ is a non obligation-producing action which is 
		an obligation-releasing one, then 
		\begin{align*}
			\begin{split}  
				{\mathcal{R}}[\neg \textbf{Oblg}(W,do(\alpha, \sigma))] = 
				{\mathcal{R}}[\mathcal{\Phi}_{W} \land
				\textbf{Oblg}(\rho^1(W,\alpha),\sigma)]. 
			\end{split}
		\end{align*} 
		\item[(iv)] If $\alpha$ is not an obligation-producing action which is     
		neither an obligation-releasing one, nor a non obligation-releasing 
		one, then 
		\begin{align*}
			\begin{split}  
				{\mathcal{R}}[\textbf{Oblg}(W,do(\alpha, \sigma))] = 
				{\mathcal{R}}[\textbf{Oblg}(\rho^1(W,\alpha),\sigma)]. 
			\end{split}
		\end{align*}                         
		\item[(v)] If $\alpha = A_{\psi}(\vec t)$ is an obligation-producing 
		action that brings about the obligation of the formula $\psi$ 
		for some term $\vec t$, then
		\begin{align*}
			\begin{split}  
				{\mathcal{R}}[\textbf{Oblg}(W,&do(\alpha, \sigma))] = 
				{\mathcal{R}}[(\exists \alpha')\textbf{Oblg}(\psi(\vec t) \supset 
				\rho^1(W,\alpha'), \sigma)].                                     
			\end{split}
		\end{align*}                
		\item[(vi)] If $\alpha = A(\vec t)$ is an obligation-producing action 
		that requires the referent of a term $t$ to be some individual $y$,
		then
		\begin{align*}
			\begin{split}  
				{\mathcal{R}}[\textbf{Oblg}(&W,do(\alpha, \sigma))] =  
				{\mathcal{R}}[(\exists y, \alpha')\textbf{Oblg}(t = y \supset
				\rho^1(W,\alpha'),\sigma)].                                     
			\end{split}
		\end{align*}          
	\end{description}
\end{definition}
In the sequel, the notation ${\mathcal{R}}^* [W]$ is used to denote a repeated application of the regression operator on the formula $W$ until further applications leave $W$ unchanged.
We need to extend the basic action theories from ~\cite{DBLP:journals/logcom/LinR94} with  new classes of axioms and revisit and/or extend some basic sentences of basic action theories for knowledge from ~\cite{DBLP:journals/tocl/Reiter01}. The extensions essentially take into account the inclusion of $\textbf{Oblg}$ and $\textbf{Oref}$ in the action precondition axioms and axioms about $S_0$, and the inclusion of the successor state axiom for $O$. With basic action theories for obligations in hand, we now have:
\begin{theorem}{\bf (The Regression Theorem with Obligations)} 
	\label{REGRESSION-THEOREM-WITH-OBLIGATIONS}     
	Suppose that $\mathcal{D}$ is a basic action theory for obligations, and $W$ is a regressable sentence. Let ${\mathcal{D}}_{S_0}$ and  
	${\mathcal{D}}_{una}$ be the axioms describing the initial situation and unique name axioms, respectively. Moreover, suppose that ${\mathcal{O}}_{Init}$ consists of the properties of secondary refexivity, and secondary seriality. 
	Then 
	\[{\mathcal{D}} \models W \hspace{.2cm}\mbox{iff} \hspace{.2cm} {\mathcal{D}}_{S_0}
	\cup {\mathcal{D}}_{una} \cup {\mathcal{O}}_{Init} \models {\mathcal{R}}^*[W].\]
\end{theorem} 
\begin{proof}  {\bf Outline}.
	The proof extends the one for the basic regression theorem given in~\cite{DBLP:conf/aaai/ScherlL93} by handling the regression steps related to the obligation. Each one of the steps {\bf (i)}-{\bf (vi)} of the extended regression operator is justified by a proposition whereby the sentence mentioned on the left-hand and right-hand sides of the equality sign are shown to be logically equivalent.That is, each step of the regression preserves logical equivalence. the whole process terminates in $S_0$ after a finite number of steps.\\[.2cm]  
	{\bf Case(i)}: The proof of the regression rule expressed by the the following \({\mathcal{R}}[\textbf{Oblg}(W,S_0)] = \textbf{Oblg}(W,S_0)\)  is immediate, since for a situation suppressed expression $W$ with free variables  $\vec x$, we have:  
	\[{\mathcal{D}}\models (\forall \vec x).\textbf{Oblg}(W(\vec x),S_0) \equiv \textbf{Oblg}(W(\vec x),S_0).\]  
	{\bf Case(ii)}: Here the regression rule is justified by the following: suppose $W$ is a situation suppressed expression  with free variables  $\vec x$, and $A(\vec y)$ is a non obligation-producing action that is a non obligation-releasing one; then  

	\[
	\begin{split}
		\mathcal{D} \models (\forall \vec{x}, \vec{y}, \vec{s})\;
		&\textbf{Oblg}\!\left(W(\vec{x}), \operatorname{do}(A(\vec{y}), s)\right)
		\;\equiv\; \\
		&\neg \Phi_{G_1} \land \dots \land \neg \Phi_{G_l}
		\land\;
		\textbf{Oblg}\!\left(\rho^{1}(W(\vec{x}), A(\vec{y})), s\right).
	\end{split}
	\]
	{\bf Case(iii)}: This regression rule is justified by the following: suppose $W$ is a situation suppressed expression  with free variables  $\vec x$, and $A(\vec y)$ is a non obligation-producing action that is an obligation-releasing one; then  
	
	\[
	\begin{split}
		\mathcal{D} \models (\forall \vec{x}, \vec{y}, \vec{s})\;
		\textbf{Oblg}(W(\vec{x})&, \operatorname{do}(A(\vec{y}), s))
		\;\equiv\; \\
		&\neg \mathcal{\Phi}_{W}(\vec{x}) \land\;
		\textbf{Oblg}(\rho^{1}(W(\vec{x}), A(\vec{y})), s).
	\end{split}
	\]
	
	{\bf Case(iv)}: Suppose $W$ is a situation suppressed expression  with free variables  $\vec x$, and $A(\vec y)$ is a non obligation-producing action that is neither an obligation-releasing one, nor a non obligation-releasing one ; then  
%
	
%
	
	\[
	\begin{split}
		\mathcal{D} \models (\forall \vec{x}, \vec{y}, s)\;
		\textbf{Oblg}(W(\vec{x})&, \operatorname{do}(A(\vec{y}), s))
		\;\equiv\; \\
		&(\exists A')\,\textbf{Oblg}(
		\psi(\vec{y}) \supset \rho^{1}(W(\vec{x}), A'),\, s
		).
	\end{split}
	\]
	
	{\bf Case(vi)}: Suppose $W$ is a situation suppressed expression  with free variables  $\vec x$, and $A_{t}(y)$ is an  obligation-producing action that requires the referent of some term $t$ to be an individual $y$ ; then  
	
	\[
	\begin{split}
		\mathcal{D} \models (\forall \vec{x}, y, s)\;
		&\textbf{Oblg}(W(\vec{x}), \operatorname{do}(A_{t}(y), s))
		\;\equiv\; \\
		&(\exists t, A')\;
		\bigl(
		\textbf{Oblg}(t = y) \supset
		\textbf{Oblg}(\rho^{1}(W(\vec{x}), A'), s)
		\bigr).
	\end{split}
	\]

\end{proof}
Theorem ~\ref{REGRESSION-THEOREM-WITH-OBLIGATIONS} states that in order to evaluate a sentence $W$ against a basic action theory $\mathcal{D}$, it is necessary and sufficient to evaluate ${\mathcal{R}}[W]$ in the initial theory ${\mathcal{D}}_{S_0}$ augmented with unique name axioms and properties of the relation $O$. 
\subsection{Example}
Suppose that the axiomatization of the initial situation includes $door(D,S_0)$, $locked(D,S_0)$, and $at(D,S_0)$. Assume we regress the sentence~\ref{LOCKING-OBLIGATION}.  
\begin{align}
	\begin{split}\label{LOCKING-OBLIGATION} 
		\textbf{Oblg}(locked(D), do(pressButton&(D,20),\\
		&do(unlock(D,10),S_0))))).  
	\end{split}
\end{align}
By the Step~{\bf iv} of Definition~\ref{OBLIGATION-REGRESSION-OPERATOR} on the sentence (\ref{LOCKING-OBLIGATION}) we obtain:
\begin{align}\label{LOCKING-OBLIGATION-2} 
	\begin{split}  
		{\mathcal{R}}[\textbf{Oblg}(\rho^1(locked(D),pressButton&(D,20)),\\
		&do(unlock(D,10),S_0))], 
	\end{split}
\end{align}    
which, by the one-step reduction, expands to
\begin{align}\label{LOCKING-OBLIGATION-3} 
	\begin{split}  
		{\mathcal{R}}[\textbf{Oblg}(locked(D), do(unlock(D,10),S_0))]. 
	\end{split}
\end{align} 
Furthermore, recall that $locked(d,s)$ is the formula that is associated  with the obligation-producing action $unlock_{locked}(d,s)$; thus performing Step~{\bf v} of Definition~\ref{OBLIGATION-REGRESSION-OPERATOR} on the sentence (\ref{LOCKING-OBLIGATION-3}) yields 
the sentence~(\ref{REGRESSED-LOCKING-OBLIGATION}). 
\begin{align}
	\begin{split}\label{REGRESSED-LOCKING-OBLIGATION} 
		{\mathcal{R}}[(\exists a').\textbf{Oblg}(&locked(D) \supset  
		((\exists t).a' = lock(D,t) \lor \\
		&\hspace{.3cm} locked(D) \land a' \not = pressButton(D,t)),S_0)]. 
	\end{split}
\end{align} 
\section{Related Work}\label{RelatedWork}
Our obligation-producing actions in the Situation Calculus are similar to and a  substantial modification of those that were first introduced  in ~\cite{DBLP:conf/ecai/Demolombe04} and in~\cite{DBLP:conf/deon/DemolombeH04},  where the deontic accessibility relationship $O$ was introduced. 
In ~\cite{DBLP:conf/ecai/Demolombe04} and~\cite{DBLP:conf/deon/DemolombeH04}, authors ranked deontic alternatives in terms of their levels of ideality; they subsequently define the obligatory sentences as those that are true in all alternative situations with maximal ideality; and, finally, they give a successor state axiom for the fluent $O$. By contrast, our work simplifies the formalization by removing any use of situation idealities and by solely embedding the possible world semantics for SDL from \cite{Hintikka70} in the Situation Calculus. Furthermore, we expand Reiter's regression to reason about obligation-producing actions.
Another approach for incorporating deontic notions into the Situation Calculus is presented by Classen and Delgrande in~\cite{DBLP:conf/kr/ClassenD20}. 
In~\cite{DBLP:conf/kr/ClassenD20}, deontic assertions and modalities are expressed as constraints that subsequently  compiled into a Situation Calculus action theory which are used to reason about obligations. We differ from this approach by expressing obligations directly in the Situation Calculus so that there is no need of an extra compilation step.  


Event Calculus described in~\cite{kowalski1989logic} may also be used for specifying obligations~(See~\cite{daskalopulu2000modelling}, and~\cite{DBLP:conf/ruleml/HashmiGW14}. The Situation Calculus, however, enjoys the key advantage of the existence of GOLOG 
~(See in~\cite{DBLP:journals/jlp/LevesqueRLLS97}), a Situation Calculus-based programming language for defining complex actions in terms of a set of primitive actions axiomatized in the Situation Calculus. 


\section{Conclusion and Future Work}\label{Conclusion}

We have spelled out the formalization of obligations in the Situation Calculus around a deontic fluent $O$, along with appropriate actions that change its truth value, as well as an appropriate successor state axiom for that deontic fluent.  Furthermore, we stated the correctness of our approach as well as appropriate restrictions of the deontic fluent in the initial situations. Finally, we extended the regression operator, the main reasoning tool of Reiter's  framework, to obligations.

 
One future work is the specification of smart legal contracts in the Situation Calculus. The plan is to expand basic action theories to {\it basic contractual theories} to provide the formal semantics of smart legal contracts to represent legal contracts as processes in the Situation Calculus. Thus, we obtain an  implementable specification that allows to automatically check many properties of the specification using an obligation-based GOLOG interpreter. \\
Another Future work will enhance the current framework by integrating mechanisms to address norm violations and contrary-to-duty (CTD) obligations. Specifically, we intend to introduce explicit violation predicates to signify the non-fulfillment of primary obligations. Concurrently, we will define secondary obligations that are triggered by these violations. This will enable the model to represent reparational duties and more accurately portray normative scenarios. Furthermore, we plan to explore the interplay between these augmented obligations and the underlying action theory of the Situation Calculus, alongside their application in legal reasoning and smart contract verification.

\printbibliography

@String{JACM = "J. ACM" }

@String{Computer = "{IEEE} Computer" }

@String{Academic = "Academic Press" }

@String{Springer = "Springer-Verlag" }

@article{LevesquePR98,
  author    = {Hector J. Levesque and
               Fiora Pirri and
               Raymond Reiter},
  title     = {Foundations for the Situation Calculus},
  journal   = {Electron. Trans. Artif. Intell.},
  volume    = {2},
  pages     = {159--178},
  year      = {1998},
  url       = {http://www.ep.liu.se/ej/etai/1998/005/},
  bibsource = {dblp computer science bibliography, https://dblp.org}
}

@article{pirri1999some,
  title={Some contributions to the metatheory of the situation calculus},
  author={Pirri, Fiora and Reiter, Ray},
  journal={Journal of the ACM (JACM)},
  volume={46},
  number={3},
  pages={325--361},
  year={1999},
  publisher={ACM New York, NY, USA}
}

@article{reiter1991frame,
  title={The frame problem in the situation calculus: A simple solution (sometimes) and a completeness result for goal regression.},
  author={Reiter, Raymond},
  journal={Artificial and Mathematical Theory of Computation},
  volume={3},
  year={1991},
  publisher={Citeseer}
}

@inproceedings{DBLP:conf/birthday/Reiter91,
  author    = {Raymond Reiter},
  editor    = {Vladimir Lifschitz},
  title     = {The Frame Problem in the Situation Calculus: {A} Simple Solution (Sometimes) and a Completeness Result for Goal Regression},
  booktitle = {Artificial and Mathematical Theory of Computation, Papers in Honor of John McCarthy on the occasion of his sixty-fourth birthday},
  pages     = {359--380},
  publisher = {Academic Press / Elsevier},
  year      = {1991},
  url       = {https://doi.org/10.1016/b978-0-12-450010-5.50026-8},
  doi       = {10.1016/b978-0-12-450010-5.50026-8},
  bibsource = {dblp computer science bibliography, https://dblp.org}
}

@article{macleod20002007,
Author = {MacLeod, W. Bentley},
Title = {Reputations, Relationships, and Contract Enforcement},
Journal = {Journal of Economic Literature},
Volume = {45},
Number = {3},
Year = {2007},
Month = {September},
Pages = {595-628},
DOI = {10.1257/jel.45.3.595},
URL = {https://www.aeaweb.org/articles?id=10.1257/jel.45.3.595}}

@article{DBLP:journals/dss/Lee88,
  author       = {Ronald M. Lee},
  title        = {A logic model for electronic contracting},
  journal      = {Decis. Support Syst.},
  volume       = {4},
  number       = {1},
  pages        = {27--44},
  year         = {1988},
  url          = {https://doi.org/10.1016/0167-9236(88)90096-6},
  doi          = {10.1016/0167-9236(88)90096-6},
  bibsource    = {dblp computer science bibliography, https://dblp.org}
}

@inproceedings{sperotto2019ontology,
  title={Ontology-based legal system in multi-agents systems},
  author={Sperotto, F{\'a}bio Aiub and Belchior, Mairon and de Aguiar, Marilton Sanchotene},
  booktitle={Mexican International Conference on Artificial Intelligence},
  pages={507--521},
  year={2019},
  organization={Springer}
}

@article{drumond2008multi,
  title={A multi-agent legal recommender system},
  author={Drumond, Lucas and Girardi, Rosario},
  journal={Artificial intelligence and law},
  volume={16},
  number={2},
  pages={175--207},
  year={2008},
  publisher={Springer}
}

@book{brown2014frame,
  title={The frame problem in artificial intelligence: Proceedings of the 1987 workshop},
  author={Brown, Frank M},
  year={2014},
  publisher={Morgan Kaufmann}
}

@incollection{hayes1981frame,
  title={The frame problem and related problems in artificial intelligence},
  author={Hayes, Patrick J},
  booktitle={Readings in artificial intelligence},
  pages={223--230},
  year={1981},
  publisher={Elsevier}
}

@article{prakken2015law,
  title={Law and logic: A review from an argumentation perspective},
  author={Prakken, Henry and Sartor, Giovanni},
  journal={Artificial intelligence},
  volume={227},
  pages={214--245},
  year={2015},
  publisher={Elsevier}
}

@article{hansson2006ideal,
  title={Ideal worlds—wishful thinking in deontic logic},
  author={Hansson, Sven Ove},
  journal={Studia logica},
  volume={82},
  number={3},
  pages={329--336},
  year={2006},
  publisher={Springer}
}

@incollection{hintikka1971some,
  title={Some main problems of deontic logic},
  author={Hintikka, Jaakko},
  booktitle={Deontic logic: Introductory and systematic readings},
  pages={59--104},
  year={1971},
  publisher={Springer}
}

@article{DBLP:journals/auasjlog/GovernatoriR06,
  author       = {Guido Governatori and
                  Antonino Rotolo},
  title        = {Logic of Violations: {A} Gentzen System for Reasoning with Contrary-To-Duty
                  Obligations},
  journal      = {Australas. J. Log.},
  volume       = {4},
  year         = {2006},
  url          = {https://doi.org/10.26686/ajl.v4i0.1780},
  doi          = {10.26686/AJL.V4I0.1780},
  bibsource    = {dblp computer science bibliography, https://dblp.org}
}

@inproceedings{DBLP:conf/deon/GovernatoriOCR16,
  author       = {Guido Governatori and
                  Francesco Olivieri and
                  Erica Calardo and
                  Antonino Rotolo},
  editor       = {Olivier Roy and
                  Allard M. Tamminga and
                  Malte Willer},
  title        = {Sequence Semantics for Norms and Obligations},
  booktitle    = {Deontic Logic and Normative Systems - 13th International Conference,
                  {DEON} 2016, Bayreuth, Germany, July 18-21, 2018},
  pages        = {93--108},
  publisher    = {College Publications},
  year         = {2016},
  bibsource    = {dblp computer science bibliography, https://dblp.org}
}

@inproceedings{DBLP:conf/ruleml/HashmiGW14,
	title={Modeling obligations with event-calculus},
	author={Hashmi, Mustafa and Governatori, Guido and Wynn, Moe Thandar},
	booktitle={International Symposium on Rules and Rule Markup Languages for the Semantic Web},
	pages={296--310},
	year={2014},
	organization={Springer}
}

@article{DBLP:journals/jacm/PirriR99,
	author       = {Fiora Pirri and
	Raymond Reiter},
	title        = {Some Contributions to the Metatheory of the Situation Calculus},
	journal      = {J. {ACM}},
	volume       = {46},
	number       = {3},
	pages        = {325--361},
	year         = {1999},
	% url          = {https://doi.org/10.1145/316542.316545},
	% doi          = {10.1145/316542.316545},
	timestamp    = {Wed, 14 Nov 2018 10:35:25 +0100},
	% url      = {https://dblp.org/rec/journals/jacm/PirriR99.bib},
	bibsource    = {dblp computer science bibliography, https://dblp.org}
}

@article{DBLP:journals/tocl/Reiter01,
	author       = {Raymond Reiter},
	title        = {On knowledge-based programming with sensing in the situation calculus},
	journal      = {{ACM} Trans. Comput. Log.},
	volume       = {2},
	number       = {4},
	pages        = {433--457},
	year         = {2001},
	% url          = {https://doi.org/10.1145/383779.383780},
	% doi          = {10.1145/383779.383780},
	timestamp    = {Tue, 04 Jun 2019 16:08:56 +0200},
	% url      = {https://dblp.org/rec/journals/tocl/Reiter01.bib},
	bibsource    = {dblp computer science bibliography, https://dblp.org}
}

@inproceedings{DBLP:conf/kr/ClassenD20,
  author       = {Jens Cla{\ss}en and
                  James P. Delgrande},
  title        = {Dyadic Obligations over Complex Actions as Deontic Constraints in
                  the Situation Calculus},
  booktitle    = {Proceedings of the 17th International Conference on Principles of
                  Knowledge Representation and Reasoning, {KR} 2020, Rhodes, Greece,
                  September 12-18, 2020},
  pages        = {253--263},
  year         = {2020},
  % url          = {https://doi.org/10.24963/kr.2020/26},
  % doi          = {10.24963/KR.2020/26},
  timestamp    = {Fri, 29 Jan 2021 19:06:04 +0100},
  % url      = {https://dblp.org/rec/conf/kr/ClassenD20.bib},
  bibsource    = {dblp computer science bibliography, https://dblp.org}
}

@inproceedings{daskalopulu2000modelling,
  title={Modelling legal contracts as processes},
  author={Daskalopulu, Aspassia},
  booktitle={Proceedings 11th International Workshop on Database and Expert Systems Applications},
  pages={1074--1079},
  year={2000},
  organization={IEEE}
}

@inbook{Hintikka70,
  author    = {Hintikka, Jaakko},
  editor    = {Risto Hilpinen},
  title     = {Some Main Problems of Deontic Logic}, 
  bookTitle = {Deontic Logic: Introductory and Systematic Readings},
  year      = {1970},
  publisher= {Springer Netherlands},
  address = {Dordrecht, Holland},
  pages     = {59--104}

}

@inproceedings{DBLP:conf/aaai/ScherlL93,
  author    = {Richard B. Scherl and
               Hector J. Levesque},
  title     = {The Frame Problem and Knowledge-Producing Actions},
  booktitle = {Proceedings of the 11th National Conference on Artif. Intell.,
               Washington, DC, July 11-15, 1993},
  pages     = {689--695},
  publisher = {{AAAI} Press},
  year      = {1993},
  % url       = {http://www.aaai.org/Library/AAAI/1993/aaai93-103.php},
  timestamp = {Tue, 11 Dec 2012 17:40:41 +0100},
  % url   = {https://dblp.org/rec/conf/aaai/ScherlL93.bib},
  bibsource = {dblp computer science bibliography, https://dblp.org}
}

@article{DBLP:journals/jlp/LevesqueRLLS97,
  author    = {Hector J. Levesque and
               Raymond Reiter and
               Yves Lesp{\'{e}}rance and
               Fangzhen Lin and
               Richard B. Scherl},
  title     = {{GOLOG:} {A} Logic Programming Language for Dynamic Domains},
  journal   = {J. Log. Program.},
  volume    = {31},
  number    = {1-3},
  pages     = {59--83},
  year      = {1997},
  % url       = {https://% doi.org/10.1016/S0743-1066(96)00121-5},
  % doi       = {10.1016/S0743-1066(96)00121-5},
  timestamp = {Wed, 17 Feb 2021 08:54:34 +0100},
  % url   = {https://dblp.org/rec/journals/jlp/LevesqueRLLS97.bib},
  bibsource = {dblp computer science bibliography, https://dblp.org}
}

@article{DBLP:journals/etai/LevesquePR98,
  author    = {Hector J. Levesque and
               Fiora Pirri and
               Raymond Reiter},
  title     = {Foundations for the Situation Calculus},
  journal   = {Electron. Trans. Artif. Intell.},
  volume    = {2},
  pages     = {159--178},
  year      = {1998},
  % url       = {http://www.ep.liu.se/ej/etai/1998/005/},
  timestamp = {Wed, 05 May 2004 09:50:02 +0200},
  % url   = {https://dblp.org/rec/journals/etai/LevesquePR98.bib},
  bibsource = {dblp computer science bibliography, https://dblp.org}
}

@incollection{kowalski1989logic,
  title={A logic-based calculus of events},
  author={Kowalski, Robert and Sergot, Marek},
  booktitle={Foundations of knowledge base management},
  pages={23--55},
  year={1989},
  publisher={Springer}
}

@book{reiter2001knowledge,
  title={Knowledge in action: logical foundations for specifying and implementing dynamical systems},
  author={Reiter, Raymond},
  year={2001},
  publisher={MIT press}
}

@inproceedings{DBLP:conf/otm/KruijffW17,
title={Ontologies for commitment-based smart contracts},
author={De Kruijff, Joost and Weigand, Hans},
booktitle={OTM 2017: Confederated International Conferences: CoopIS, C\&TC, and ODBASE 2017, Rhodes, Greece, October 23-27, 2017, Proceedings, Part II},
pages={383--398},
year={2017},
organization={Springer}
}

@inproceedings{DBLP:conf/deon/DemolombeH04,
title={Obligation change in dependence logic and situation calculus},
author={Demolombe, Robert and Herzig, Andreas},
booktitle={7th International Workshop on Deontic Logic in Computer Science, DEON 2004, Madeira, Portugal, May 26-28, 2004. Proceedings 7},
pages={57--73},
year={2004},
organization={Springer}
}

@book{vonWright1951-VONAEI-2,
address = {Amsterdam},
author = {G. H. von Wright},
editor = {},
publisher = {North-Holland Pub. Co.},
title = {An Essay in Modal Logic},
year = {1951}
}

@incollection{moore1981reasoning,
title={Reasoning about knowledge and action},
author={Moore, Robert C},
booktitle={Readings in Artif. Intell.},
pages={473--477},
year={1981},
publisher={Elsevier}
}

@techreport{mccarthy1963situations,
title={Situations, actions, and causal laws},
author={McCarthy, John},
year={1963},
institution={Stanford University, Dept of Computer Science}
}

@inproceedings{DBLP:conf/ecai/Demolombe04,
author    = {Robert Demolombe},
title     = {From Belief Change to Obligation Change in the Situation Calculus},
booktitle = {Proceedings of the 16th Eureopean Conference on Artif. Intell.,
ECAI'2004, {PAIS} 2004, Valencia, Spain, August 22-27, 2004},
pages     = {991--992},
publisher = {{IOS} Press},
year      = {2004},
timestamp = {Fri, 11 May 2018 12:42:30 +0200},
% url   = {https://dblp.org/rec/conf/ecai/Demolombe04.bib},
bibsource = {dblp computer science bibliography, https://dblp.org}
}

@article{DBLP:journals/logcom/LinR94,
author    = {Fangzhen Lin and
Raymond Reiter},
title     = {State Constraints Revisited},
journal   = {J. Log. Comput.},
volume    = {4},
number    = {5},
pages     = {655--678},
year      = {1994},
% url       = {https://% doi.org/10.1093/logcom/4.5.655},
% doi       = {10.1093/logcom/4.5.655},
timestamp = {Wed, 17 May 2017 14:25:56 +0200},
% url   = {https://dblp.org/rec/journals/logcom/LinR94.bib},
bibsource = {dblp computer science bibliography, https://dblp.org}
}
\end{document}